\documentclass[runningheads,a4paper]{llncs}
\usepackage{amssymb}
\usepackage{graphicx}
\usepackage{amsmath}

\usepackage{url}

\begin{document}

\title{On computing optimal Locally Gabriel Graphs}

\titlerunning{On Computing Optimal Locally Gabriel Graphs}

%
%
\author{
Abhijeet Khopkar
\and
Sathish Govindarajan
}
\authorrunning{Abhijeet Khopkar and Sathish Govindarajan}

\institute{Department of Computer Science and Automation\\
Indian Institute of Science, Bangalore}
%
%

\maketitle
\begin{abstract}
Delaunay and Gabriel graphs are widely studied geometric proximity structures. Motivated by applications in
wireless routing, relaxed versions of these graphs known as \emph{Locally Delaunay Graphs} ($LDGs$) and \emph{Locally Gabriel Graphs} ($LGGs$) were proposed.
We propose another generalization of $LGGs$ called \emph{Generalized Locally Gabriel Graphs} ($GLGGs$) in the context when certain edges are forbidden in the graph.
Unlike a Gabriel Graph, there is no unique $LGG$ or $GLGG$ for a given point set because no edge is necessarily included or excluded.
This property allows us to choose an $LGG/GLGG$ that optimizes a parameter of interest in the graph. 
We show that computing an edge maximum $GLGG$ for a given problem instance is NP-hard and also APX-hard. 
We also show that computing an $LGG$ on a given point set with dilation $\le k$ is NP-hard. Finally, we give an algorithm to verify whether a given geometric graph $G=(V,E)$ is a valid $LGG$.
\end{abstract}

\section{Introduction}
A geometric graph $G=(V,E)$ is an embedding of the set $V$ as points in the plane and the set $E$ as line segments joining two points in $V$.
 Delaunay graphs, Gabriel graphs and Relative neighborhood graphs (RNGs) are classic examples of
geometric graphs that have been extensively studied and have applications in computer graphics, GIS, wireless networks, sensor networks, etc (see survey \cite{JT}).
Gabriel and Sokal \cite{ggg} defined the Gabriel graph as follows:
\begin{definition} \label{gdef}
A geometric graph $G=(V,E)$ is called a Gabriel graph if the following condition holds:
For any $u,v \in V$, an edge $(u,v) \in E$ if and only if the circle with $\overline{uv}$ as diameter does not contain any other point of $V$.
\end{definition}

Gabriel graphs have been used to model the topology in a wireless network \cite{Bose,Urut}. Motivated by applications in wireless routing, Kapoor and Li~\cite{yyy} proposed a relaxed version of Delaunay/Gabriel graphs known as $k$-locally Delaunay/Gabriel graphs.
The edge complexity of these structures has been studied in \cite{yyy,ps04}.
In this paper, we focus on 1-locally Gabriel graphs and call them {\em Locally Gabriel Graphs} ($LGG$s).
\begin{definition}
A geometric graph $G=(V,E)$ is called a Locally Gabriel Graph if for every $(u,v) \in E$, the circle with $\overline{uv}$ as diameter
does not contain any neighbor of $u$ or $v$ in $G$.
\end{definition}
The above definition implies that in an $LGG$, two edges $(u,v) \in E$ and $(u,w) \in E$ {\it conflict} with each other and cannot co-exist 
if $\angle uwv~\ge~\frac{\pi}{2}$ or $\angle uvw \ge \frac{\pi}{2}$. Conversely if edges $(u,v)$ and $(u,w)$ co-exist in an $LGG$,
then $\angle uwv~<~\frac{\pi}{2}$ and $\angle uvw < \frac{\pi}{2}$. We call this condition as {\it LGG constraint}.

Study of these graphs was initially motivated by design of dynamic routing protocols for \emph {ad hoc} wireless networks \cite{Li}. 
An ad-hoc wireless network consists of a collection of wireless transceivers (corresponds to the the points) and an underlying network topology (corresponds to the edges)
that is used for communication.
Like Gabriel Graphs, $LGGs$ are also proximity based structures that capture the interference patterns in wireless networks.
An interesting point to be noted is that there is no unique $LGG$ on a given point set since no edge in an $LGG$ is necessarily included or excluded.
Thus the edge set of the graph (used for wireless communication) can be customized to optimize certain network parameters depending on the application.
While a Gabriel graph has linear number of edges (planar graph), an $LGG$ can be constructed with super-linear number of edges \cite{Erdudg}. A dense network can be desirable for applications like
broadcasting or multicasting where a large number of pairs of nodes need to communicate with each other and links have limited bandwidth.
The dilation or spanning ratio of a graph is an important parameter in wireless network design. It is the maximum ratio of the distance in the network (length of the shortest path)
to the Euclidean distance for any two nodes in a wireless network. Graphs with small spanning ratios are important in many applications and motivate the
study of geometric spanners (refer to \cite{Epp} for a survey). Proximity graphs have been studied for their dilation.
Some interesting bounds for the dilation of Gabriel Graphs were presented in \cite{beta}. In this paper, we initiate study for dilation on $LGGs$.
We show that for certain point sets there exist $LGGs$ with $O(1)$ dilation whereas the Gabriel graph on the same point set has dilation $\Omega(\sqrt{n})$.

In many situations, certain links are forbidden in a network due to physical barriers, visibility constraints or
limited transmission radius. Thus, all proximate pairs of nodes might not induce edges and this effect can be considered in $LGGs$.
Thus, it is natural to study $LGGs$ in the context when the network has to be built only with a set of predefined links. 
In this context, we define a generalized version of $LGGs$ called \emph{Generalized locally Gabriel Graphs}
($GLGGs$). Edges in a $GLGG$ can be picked only from the edges in a given predefined geometric graph.
\begin{definition}
For a given geometric graph $G=(V,E)$ we define $G'=(V,E')$ as GLGG if $G'$ is a valid LGG and $E' \subseteq E$.
\end{definition}

Previous results on $LGGs$ have focused on obtaining combinatorial bounds on the maximum edge complexity. In \cite{yyy}, it was shown that an $LGG$
has at most $O(n^\frac{3}{2})$ edges since $K_{2,3}$ is a forbidden subgraph. Also, it was observed in \cite{ps04} that any unit distance graph is also a valid $LGG$.
Hence there exist $LGGs$ with $\Omega(n^{1+\frac{c}{\log\log n}})$ edges \cite{Erdudg}.
It is not known whether an edge maximum $LGG$ can be computed in polynomial time.
\newpage
\paragraph{Our Contribution:}
We present the following results in this paper.
\begin{enumerate}
\item
We show that computing a $GLGG$ with at least $m$ edges on a given geometric graph $G=(V,E)$ is NP-complete (reduction from 3-SAT) and also APX-hard (reduction from MAX-(3,4)-SAT).
\item
We show that the problem of determining whether there exists an $LGG$ with dilation $\leq k$ is NP-hard by reduction from the partition problem motivated by \cite{nph}.  
We also show that there exists a point set $P$ such that any $LGG$ on $P$ has dilation $\Omega(\sqrt{n})$ that matches with the best known upper bound \cite{beta}.
\item
For a given geometric graph $G=(V,E)$, we give an algorithm with running time $O(|E|\log |V|+|V|)$ to verify whether $G$ is a valid $LGG$.
\end{enumerate}
\section{Hardness of computing an edge maximum $GLGG$}
In this section we show that deciding whether there exists a $GLGG$ on a given geometric graph $G=(V,E)$ with at least $m$ edges for a given value of $m$ is NP-complete by a reduction from 3-SAT.
We further show that computing edge maximum $GLGG$ is APX-hard by showing a reduction from MAX-(3,4)-SAT.

A 3-SAT instance is a conjunction of several clauses and each clause is a disjunction of exactly 3 variables.
Let $\mathcal I$ be an instance of the 3-SAT problem with $k$ clauses $C_1, C_2, \ldots ,C_k$ and $n$ variables $y_1, y_2, \ldots ,y_n$.
A geometric graph $G~=~(V,E)$ is constructed from $\mathcal{I}$ such that there exists a $GLGG$ on $G$ with at least $m$ edges if and only if $\mathcal{I}$ admits a
satisfying assignment.
We construct a vertex set $V$ (points in the plane) of size $(k+3)n + k$ that is partitioned into 
$2n$ literal vertices denoted by $V_1=\{x_i, x_i'\ |\ i \in \{1, \ldots, n\}\}$, $(k+1)n$ variable vertices denoted by $V_2=\{z_{i_j} \ |\ i \in \{1, \ldots, n\}, j \in \{1, \ldots, k+1\}\}$
and $k$ clause vertices denoted by $V_3=\{c_j \ |\ j \in \{1, \ldots, k\}\}$. Thus, $V = V_1 \cup V_2 \cup V_3$.
Two literal vertices $x_i$ and $x'_i$ corresponding to the same variable are called conjugates of each other.

Now let us discuss the placement of these vertices on the plane as shown in Figure~\ref{rfig1}.
All literal vertices are placed closely on a vertical line $l$ and the distance between two consecutive vertices is $10^{-5}$. 
Two conjugate literal vertices corresponding to the same variable are kept next to each other. Let $l_1$ and $l_2$ be two
horizontal lines passing through the highest and the lowest literal vertex respectively. Let $b_0$ be the center point of the line segment containing all the literal vertices.
With $b_0$ as center, a circle is drawn with radius $d_1 = n^4$. All clause vertices $c_1,c_2,\ldots,c_k$ are placed along an arc $a_0$ of this circle (with a distance of $\frac{n}{2}$ 
between two consecutive vertices) with the additional restriction
that these vertices cannot lie between lines $l_1$ and $l_2$. $b_0c_1$ and $b_0c_k$ make an angle less than $\alpha = \frac{\pi}{4}$ with the horizontal axis. 
Now $k+1$ variable vertices are placed for each variable in the 3-SAT instance.
Consider two horizontal lines $l_{x_i}$ and $l_{x'_i}$ passing through literal vertices $x_i$ and $x'_i$. 
With center at the mid point of $x_i$ and $x'_i$ (call it $b_i$)
a circle is drawn with radius $d_2 = 10n^4$. Variable vertices are placed on an arc $a_i$ of this circle on the same side of $l$ where clause vertices are placed. These
vertices $z_{i_1},\ldots, z_{i_{(k+1)}}$ are placed a distance of $\frac{n}{4}$ apart with the restriction that no vertex should be placed between 
lines $l_{x_i}$ and $l_{x'_i}$. Any line connecting these vertices with $x_i$ and $x'_i$ makes an angle less than $\alpha$ with the horizontal axis. For all the variables in $\mathcal{I}$, 
corresponding variable vertices are placed similarly. For simplicity variable vertices are shown corresponding to only one variable in Figure~\ref{rfig1}.
\begin{figure}[ht]
 \centering
\includegraphics[scale=0.65]{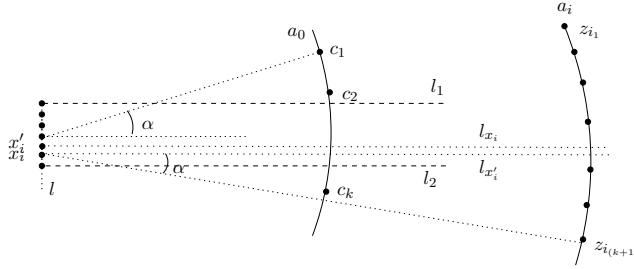}
 \caption{Placement of vertex set $V$}
 \label{rfig1}
\end{figure}
For each clause $C_j$, there are 3 edges between clause vertex $c_j$ and the corresponding literal vertices. Let $E_1$ be the set of these edges from all the clause vertices to three corresponding literal vertices.
For example, if a clause $C_j$ has literals $y_a, y_b$ and $y'_c$,
then the edges $(c_j,x_a),(c_j,x_b)$ and $(c_j,x'_c)$ are included in $E_1$. 
Another set of edges between literal vertices and variable vertices is defined
\begin{eqnarray*}
 E_2 = \{(x_i,z_{i_1}), \ldots, (x_i,z_{i_{k+1}}), (x'_i,z_{i_1}), \ldots, (x'_i,z_{i_{k+1}}) \\
~~~~~~~~~~~~~~~ |\ 1 \le i \le n\}
\end{eqnarray*}

Now, $E = E_1 \cup E_2$. Let $G = (V,E)$ be the geometric graph over which an edge maximum $GLGG$ is to be computed.
Let us analyze the conflicts among the edges in $G$. It should be noted that since a $GLGG$ is also an $LGG$, it
suffices to look at the $LGG$ constraints to determine whether two edges conflict. Consider any $GLGG$ $G'=(V,E')$ with $E' \subseteq E$ on the
geometric graph $G$. The following constraints are observed on the edge set $E'$.

 Since the edges $(x_i,z_{i_j})$ and $(x'_i,z_{i_j})$ conflict with each other 
($\angle z_{i_j}x_ix'_i$ or $\angle z_{i_j}x'_ix_i$ is greater than $\frac{\pi}{2}$ by construction),
a variable vertex $z_{i_j}$ can have an edge incident to only $x_i$ or $x'_i$. 
\begin{remark}\label{rem1}
A variable vertex $z_{i_j}$ can have only one edge ($(x_i,z_{i_j})$ or $(x'_i,z_{i_j})$) incident to it in $E'$.
\end{remark}
Similarly, we can infer Remark~\ref{rem2} due to $LGG$ constraints.
\begin{remark}\label{rem2}
 Any clause vertex $c_j$ can be incident to at most one literal vertex in $E'$.
\end{remark}
It can be observed that two $LGG$ edges that are the radii of the same circle do not conflict with each other.
Here, $b_i$ (the center of arc $a_i$) is close enough to both the literal vertices ($x_i$ and $x'_i$) and the radius $d_2$ is 
chosen large enough so that no two edges from a literal vertex to the corresponding variable vertices conflict with each other.
\begin{remark}\label{rem3}
 A literal vertex $x_i$ (or $x'_i$) can have edges incident to all the corresponding variable vertices $z_{i_j}$ in $E'$ where $j \in \{1,\ldots,k+1\}$.
\end{remark}
Since a literal vertex is placed sufficiently close to $b_0$ (the center of arc $a_0$) and the radius $d_1$ is
chosen large enough, no two edges from a literal vertex to the clause vertices conflict with each other. 
\begin{remark}\label{rem4}
In $E'$, a literal vertex $x_i$ can have edges incident to all the clause vertices that have edges incident to $x_i$ in $E_1$.
\end{remark}
Since $d_2$ is chosen large enough compared to $d_1$, if a literal vertex $x_i$ is connected
to a variable vertex $z_{i_j}$, the circle with $\overline{x_i z_{i_j}}$ as diameter would contain all the clause vertices. 
Therefore, $x_i$ cannot be connected to any clause vertex due to the $LGG$ constraint.
\begin{remark} \label{rem5}
In $E'$, if a literal vertex has an edge incident to a variable vertex, it cannot have an edge incident to any clause vertex.
\end{remark}

\begin{lemma} \label{lem3}
If there exists a $GLGG$ $G'$ on $G$ with at least $(k+1)n+k$ edges, then there exists a satisfying assignment to the given 3-SAT instance $\mathcal{I}$.
\end{lemma}
\begin{proof}
Since each variable vertex can have at most one edge incident to it (refer to Remark~\ref{rem1}), at most $(k+1)n$ edges of $E'$ can be selected from $E_2$.
Similarly each clause vertex can have at most one edge incident to it (refer to Remark~\ref{rem2}),
so in $E'$ at most $k$ edges can be selected from $E_1$. If there are $(k+1)n+k$ edges in $E'$, then one edge is incident to each variable vertex and clause vertex.
If there is an edge between a clause vertex $c_j$ and the literal vertex $x_i$ $(\text{resp. }x'_i)$, assign $y_i = 1$ $(\text{resp. }y_i = 0)$ as it satisfies the clause $C_j$. By this rule assign a truth value to a variable in each clause.
If one clause vertex is incident to $x_i$, no other clause vertex can be incident to $x'_i$ as $x'_i$ is connected to the corresponding $k+1$ variable vertices (refer to Remark~\ref{rem5}). Therefore, this
rule would yield a consistent assignment satisfying all the clauses. Hence, the given 3-SAT instance $\mathcal{I}$ is satisfiable.
\end{proof}

\begin{lemma} \label{lem4}
If there is a satisfying assignment to the given 3-SAT instance $\mathcal{I}$, then there exists a $GLGG$ $G'$ over $G$ with at least $(k+1)n+k$ edges.
\end{lemma}
\begin{proof}
A $GLGG$ with $(k+1)n+k$ edges can be constructed based on the satisfying assignment to $\mathcal{I}$. If a variable 
$y_i = 1$ $(\text{resp. }y_i = 0)$ then connect
$x'_i$ $(\text{resp. }x_i)$ to the corresponding $k+1$ variable vertices $(z_{i_1}, z_{i_2},\ldots, z_{i_{k+1}})$. Applying this rule to each variable we get $(k+1)n$
edges in $E'$ from $E_2$ and these edges do not conflict with each other (refer to Remark~\ref{rem3}). Since all the clauses will have at least one literal satisfied in this assignment, every clause vertex can have an edge incident to some literal vertex
that has no edges incident to any of the variable vertices.
Consider a clause $C_j$ which is satisfied by the assignment $y_i = 1$ (resp. $y_i = 0$). Add the edge $(c_j,x_i)$ (resp. $(c_j,x_i')$) to $E'$. Since all the clauses are satisfied, $k$ edges from $E_1$ can be added to $E'$.
Therefore, $G'$ has $(k+1)n+k$ edges and it is a valid $GLGG$.
\end{proof}

\begin{theorem}
Deciding whether there exists a $GLGG$ with at least $m$ edges for a given value of $m$ is NP-complete.
\end{theorem}
\begin{proof}
By Lemma~\ref{lem3} and Lemma~\ref{lem4}, this problem is NP-hard. Given a geometric graph $G'$, it can be verified in polynomial time whether
$G'$ is a valid $GLGG$ with at least $m$ edges. Thus, this problem is NP-complete.
\end{proof}

This reduction to argue NP-hardness can be extended further to show inapproximability for computing an edge maximum $GLGG$.
Let us consider the optimization version of 3-SAT known as MAX-3-SAT. Here the objective is to find a binary assignment satisfying the maximum number of clauses.
MAX-(3,4)-SAT is a special case of MAX-3-SAT with an additional restriction that a variable is present in exactly
four clauses. MAX-(3,4)-SAT is shown to be $APX$-hard in \cite{apx}.

Now we enhance our existing construction such that for each variable 
there are 5 variable vertices instead of $k+1$ as described in the previous reduction. Let $G=(V,E)$ be this new geometric graph on which an optimal $GLGG$ has to be computed.
Again edge sets $E_1$ and $E_2$ are defined as earlier. Now, we present the following lemma that helps to prove that computing an edge maximum $GLGG$ is $APX$-hard.

\begin{lemma} \label{apxlem}
If a $GLGG$  $G_1'$ computed over $G$ has less than $5n$ edges from $E_2$ then we can obtain another $GLGG$ $G_2'$ over $G$ with $5n$ edges from $E_2$ and $|E_2'| \ge |E_1'|$.
\end{lemma}
\begin{proof}
Initially let $G_2' = G_1'$. In $G_2'$ if a variable vertex $z_{i_j},\ 1 \le j \le 5$ has an edge incident to an associated literal vertex $x_i$, then $x_i$ cannot have an edge
incident to a clause vertex (refer to Remark~\ref{rem5}). Now $x_i$ can have edges incident to all the five variable vertices (refer to Remark~\ref{rem3}).
Therefore, if a variable vertex $z_{i_j}$ has an edge incident to $x_i$ and some other variable vertex $z_{i_{j'}}$ corresponding to the same variable
has no edge incident to it, then an edge $(x_i,z_{i_{j'}})$ can be added to $E_2'$ without conflicting with any existing edge.

If no vertex $z_{i_j}, \ 1 \le j \le 5$ has an edge incident to $x_i$, the solution can be improved locally. Add the edges $\{(x_i,z_{i_j}) | 1 \le j \le 5\}$ to $E_2'$
and remove any edges connecting $x_i$ to the clause vertices from $E_2'$. Note that a variable occurs only in four
clauses in a MAX-(3,4)-SAT instance, so a literal vertex cannot have edges incident to more than four clause vertices. Therefore, this transformation implies $|E_2'| \ge |E_1'|$.
Applying this argument to all the variable vertices, it can be ensured that in $G_2'$ every variable vertex has an edge incident to it. Thus, $E_2'$ has $5n$ edges from $E_2$ and $|E_2'| \ge |E_1'|$.
\end{proof}

\begin{theorem}
Computing an edge maximum $GLGG$ on a given geometric graph $G=(V,E)$ is $APX$-hard.
\end{theorem}
\begin{proof}
Let $OPT_G$ and $OPT_S$ denote the optimum for the $GLGG$ instance and the MAX-(3,4)-SAT instance respectively. A clause vertex can have
only one edge incident to it (refer to Remark~\ref{rem2}) and 
a $GLGG$ maximizing the edges will have $5n$ edges from $E_2$ (edges between variables vertices and literal vertices, refer to Lemma~\ref{apxlem}). Therefore, $OPT_G = 5n + OPT_S$.
Let an algorithm maximizing the number of edges selects $m$ edges from $E_1$ (edges between clause vertices and literal vertices) along with $5n$ edges from $E_2$.
Each of these $m$ edges implies a satisfied clause in the original MAX-(3,4)-SAT instance. 
Since MAX-(3,4)-SAT cannot be approximated beyond 0.99948 \cite{apx}, $m < 0.99948*OPT_S$. Let $c$ be the best approximation bound for the edge maximum $GLGG$.
Therefore, $c \le \frac{5n + 0.99948*OPT_S}{5n+OPT_S}$.
Since any binary assignment or its complement would necessarily satisfy at least half of the clauses in any given 3-SAT formula,
$OPT_S \ge \frac{k}{2}$. Here $n = \frac{3}{4}k$ implying $c \le 0.999939$. Thus, it is NP-hard to approximate edge maximum $GLGG$ 
within a factor of 0.999939.
\end{proof}

Consider the {\em maximum weight LGG} problem where the edges are assigned weights and we have to compute an 
$LGG$ maximizing the sum of the weights of the selected edges. The edge maximum $GLGG$ problem
is a special case of the maximum weight LGG problem (edge weights are either 0 or 1).
\begin{corollary}
Computing a maximum weight $LGG$ is $APX$-hard.
\end{corollary}

\section{Dilation of $LGG$}
Let us define dilation of a geometric graph $G=(V,E)$. Let $D_G(u,v)$ be the distance between two vertices in the geometric graph (sum of length of the edges in the shortest path)
and $D_2(u,v)$ be the Euclidean distance between $u$ and $v$. Let $\delta(u,v) = \frac{D_G(u,v)}{D_2(u,v)}$. The dilation of $G$ is defined as $\delta(G) = \max_{u,v \in V, u \ne v} \delta(u,v)$.
In this section, we focus on computational and combinatorial questions on dilation for $LGGs$.
\subsection{Computation of a minimum dilation $LGG$}
In this section we show that the problem of determining whether there exists an $LGG$ on a given point set with dilation $\le 7$
is NP-hard. The reduction from the partition problem is motivated by a technique in \cite{nph},
where it was shown that computing the minimum dilation geometric graph with bounded number of edges is NP-hard. Since our problem requires
us to construct an $LGG$ instead of any geometric graph with bounded number of edges, the construction needs
to be substantially modified.

The partition problem is defined as follows: Given a set $S$ of positive integers $r_i, 1\le i \le s$ s.t.
$\sum_{r \in S}  = 2R$, can it be partitioned into two disjoint sets $S_1$ and $S_2$ such that $\sum_{r \in S_1}r  = \sum_{r \in S_2}r = R$?
Given an instance of the partition problem, we construct a point set $V$ such that the instance of the partition problem is a $yes$ instance if and only if
there exists an $LGG$ on $V$ with dilation $\le 7$.
Define a parameter $\lambda$ s.t. $2sr^2_{max} < 10^\lambda$ where $r_{max}$ is the largest element of $S$.
For each $r_i \in S$, there is a gadget $G_i$. Define a parameter $\eta_i = 10^{-(\lambda+1)}r_i$ to be used in gadget $G_i$. Note that $\eta_i \le \frac{1}{10}$.

\begin{figure}[ht]
\begin{minipage}[b]{0.4\linewidth}
\centering
\input{dilm1.pstex_t}
\caption{Structure of a basic gadget}
\label{fig2}
\end{minipage}
\hspace{0.5cm}
\begin{minipage}[b]{0.5\linewidth}
\centering
\input{dilm2.pstex_t}
\caption{Basic frame structure}
\label{fig3}
\end{minipage}
\end{figure}
Now we explain the structure of a gadget $G_i$.
Each gadget comprises of 9 points as shown in Figure~\ref{fig2}. Points $x_i$ and $y_i$ are placed $1+2\eta_i$ distance apart.
$x_{i_1}$ and $y_{i_1}$ are placed at the same distance such that $\overline{x_ix_{i_1}}$ and $\overline{y_iy_{i_1}}$ are parallel to each other
and perpendicular to $\overline{x_{i_1}y_{i_1}}$.
Vertex $z_i$ is placed at the midpoint of the line segment $\overline{x_{i_1}y_{i_1}}$. $\overline{x_{i_1}x_{i_3}}$ is perpendicular to $\overline{x_ix_{i_1}}$ and
distance of $x_{i_1}$ from $x_{i_3}$ is $10\eta_i$. For $\epsilon_1 = \frac{10^{-3}}{s^2 10^{2\lambda}}$,  $x_{i_2}$ and $x_{i_3}$ are placed at a distance 
of $c_1\epsilon_1$ along $x$-axis and  $c_2\epsilon_1$ along $y$-axis for suitable constants $c_1$ and $c_2$,
s.t. $\angle x_{i_1}x_{i_2}x_{i_3} \ge \frac{\pi}{2}$. Vertices $y_{i_2}$ and $y_{i_3}$ are placed similarly.
We call edges $(x_{i_3},x_{i_2}),(x_{i_2},x_{i_1}),(x_{i_1},z_i),(z_i,y_{i_1}),(y_{i_1},y_{i_2})$ and $(y_{i_2},y_{i_3})$ {\em basic edges}. It can be verified
that an $LGG$ over the vertices of a gadget must contain all the basic edges to keep dilation bounded by 7.
It can be observed that any other edge will conflict with at least one basic edge with the exception that
the point $x_i$ can be connected to $y_i,x_{i_1}$ or $x_{i_3}$ and similarly $y_i$ can be connected to $x_i,y_{i_1}$ or $y_{i_3}$. Edges $(x_i,x_{i_1})$ and $(y_i,y_{i_1})$ are called {\em vertical edges}
 while $(x_i,x_{i_3})$ and $(y_i,y_{i_3})$ are called {\em slanted edges}. Note that the vertical edge and the slanted edge emerging from the same point $x_i$ or $y_i$
conflict with each other in an $LGG$. Additional points to be described later will ensure that there cannot exist a direct edge between $x_i$ and $y_i$. 
Though both vertices $x_i$ and $y_i$ can have independently either a vertical or a slanted edge
incident to them, if both vertices have slanted edges then $\delta(x_i,y_i) > 7$.

\begin{remark} \label{gadlem}
In a gadget $G_i$, there can be only one slanted edge if $\delta(x_i,y_i) \le 7$.
\end{remark}

\begin{figure}[!h]
\centering
\input{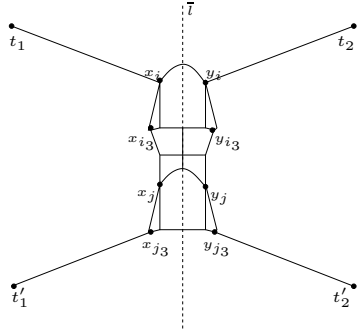}
\caption{Layout of complete structure for $s = 2$}
\label{fig4}
\end{figure}
A frame $F_i$ is used to connect two gadgets $G_i$ and $G_{i+1}$ as shown in Figure~\ref{fig3}. It connects $G_i$ at vertices $x_{i2}$ and $y_{i2}$ and
connects $G_{i+1}$ at vertices $x_{i+1}$ and $y_{i+1}$.
A frame also provides two symmetric paths ($(x_{i+1},x'_i,x_{i_2})$ and $(y_{i+1},y'_i,y_{i_2})$) between two consecutive gadgets.
Let us denote this path length between $i^{th}$ and $i+1^{th}$ gadget as $p_{i,i+1}$. All edges shown in the figure are part of the basic skeleton of a frame and
these edges are included in the set of basic edges. Here we use a technique of placing vertices at very short distance ($0.01$ in our construction) 
from each other along a line. The purpose of this technique is to ensure that all these small edges are selected in the $LGG$. If such an edge is not selected
then any alternate path shall not provide spanning ratio within limit.  We call this technique {\it vertex closing}. It will ensure that in a frame, edges 
are taken only according to our layout. Such a sequence of vertices is called a {\it vertex chain}.
An additional \emph{auxiliary vertex} is placed in each gadget $G_i$ at a distance of $\frac{\epsilon_1\eta_i}{10s}$ from $x_{i2}$ and $y_{i2}$ along the lines $\overline{x_{i2}x'_i},
\overline{x_{i_2}x_{i_1}},\overline{y_{i_2}y'_i}$ and $\overline{y_{i_2}y_{i_1}}$.

A frame also provides a convex cap on $(x_i,y_i)$ in a gadget $G_i$. This is a convex point set with all the points above $\overline{x_iy_i}$ (it need not be a regular curve).
There is a small edge incident to both $x_i$ and $y_i$ from this cap conflicting with the edge $(x_i,y_i)$ and it ensures that $x_i$ and $y_i$ are not directly connected by an edge.
It provides a path between $x_i$ and $y_i$ with spanning ratio just above 7 and for any other pair of vertices in it spanning ratio is bounded by 7. 
On the first gadget $G_1$, such a cap is placed explicitly as shown in Figure~\ref{fig4}.
Now the full structure is constructed as shown in Figure~\ref{fig4}. There is a central vertical line $\overline{l}$, all gadgets are placed along it
keeping vertex $z$ of a gadget on $\overline{l}$ s.t. $\overline{x_{i_1}y_{i_1}}$ is perpendicular to $\overline{l}$ and a frame is placed between two gadgets.
The vertical span for a frame $F_i$ is $\frac{25}{4}$.
There is a total of four extended arms, each of length $h$ with {\it vertex closing} from $G_1$ and $G_s$, each making an angle
$\sin^{-1} (\frac{220}{221})$ w.r.t. $\overline{l}$ (refer to Figure~\ref{fig4}). Here,
\[h = \frac{221}{148}(18s+ (s-1)\frac{175}{4})-\frac{k}{2}+\frac{1}{2}10^{-\lambda}R-\frac{1}{2}10^{-2\lambda}sr^2_{max}\]
where $k = \sum_{i=1}^{s-1}p_{i.i+1}+10\sum_{i=1}^{s}\eta_i$.

Let $V$ be the set of all points introduced above, clearly $|V| = O(s)$. It can be verified that the description complexity of point set $V$ is
polynomial in the size of the partition instance.

\begin{lemma} \label{one}
If the partition problem $S$ is solvable then there exists an $LGG$ on $V$ with dilation not exceeding 7.
\end{lemma}
\begin{proof}
 Let $S_1$ and $S_2$ give a partition of set $S$.
 Consider an $LGG$ on the point set $V$ obtained from $S$ that consists of all the basic edges. Additionally, in a gadget $G_i$ if $i \in S_1$ then the edges $(x_i,x_{i_3})$ and $(y_i,y_{i_1})$ are included in the $LGG$.
 Similarly, in a gadget $G_j$ if $j \in S_2$ then the edges $(y_i,y_{i_3})$ and $(x_i,x_{i_1})$ are included in the $LGG$. It is easy to verify that the spanning ratio for all pairs of points is bounded properly. 
 The only non trivial case is for $(t_1,t'_1)$ and $(t_2,t'_2)$.
 Now, we show that if there is a solution to the partition problem, then spanning ratio between $t_1$ and $t'_1$ is bounded by 7.\\
 Let $\mu_i = 10\eta_i-|x_{i_1}x_{i_2}|$ and $\nu_i = |x_{i_2}x_{i_3}|$.
 \[d_G(t_1,t'_1) = 2h+3s+\sum_{i=1}^{s-1}p_{i.i+1}+10\sum_{i=1}^{s}\eta_i - \sum_{i=1}^{s}\mu_i -\sum_{r_i \in S_1}(3+10\eta_i-\sqrt{3^2+(10\eta_i)^2}-\nu_i)\]
 Considering that $\nu_i$ is sufficiently small and $k = \sum_{i=1}^{s-1}p_{i.i+1}+10\sum_{i=1}^{s}\eta_i$
 \[\le 2h+3s+k-\sum_{r_i \in S_1}(10\eta_i - (10\eta_i)^2)\]
 \[= 2h+3s+k-\sum_{r_i \in S_1}10^{-\lambda}r_i+\sum_{r_i \in S_1}10^{-2\lambda}r_i^2\]
 \[\le 2h+3s+k-10^{-\lambda}R+10^{-2\lambda}sr^2_{max}\]
 \[=7(3s+(s-1)\frac{25}{4})+\frac{294}{221}h\]
 \[=7d_2(t_1t'_1)\]
 Symmetrically, we can argue $\delta(t_2,t'_2) \le 7$.
\end{proof}

\begin{lemma} \label{two}
 If there exists an $LGG$ on $V$ with dilation less than or equal to 7 then there exists a solution for the partition problem over $S$.
\end{lemma}
\begin{proof}
 Let us assume there exists an $LGG$ with dilation less than or equal to 7 and the partition problem has no solution. Let us first consider the case when
 the graph is composed of only basic edges and two additional edges (vertical/slanted edges from $x_i$ and $y_i$) in a gadget $G_i$. By Remark~\ref{gadlem}, there can be at most one slanted edge in any gadget.
 In this case it can be shown that $\delta(t_1,t'_1) > 7$ or $\delta(t_2,t'_2) > 7$.
 Let us assume there exists an $LGG$ with dilation less than or equal to 7 and the partition problem has no solution. Let $S_1$ be the set
 of all $r_i \in S$ such that the corresponding gadget has left edge slanted. Let $S_2 = S \setminus S_1$. Lets assume that $S_1$ and $S_2$ are not solutions
 of the original partition problem. Then w.l.o.g. it can be assumed that $\sum_{r_i \in S_1}r_i \le R -1$.\\
 Let $\mu_i = 10\eta_i-|x_{i_1}x_{i_2}|$ and $\nu_i = |x_{i_2}x_{i_3}|$.
 \[d_G(t_1,t'_1) = 2h+3s+\sum_{i=1}^{s-1}p_{i.i+1}+10\sum_{i=1}^{s}\eta_i-\sum_{i=1}^{s}\mu_i-\sum_{r_i \in S_1}(3+10\eta_i-\sqrt{3^2+(10\eta_i)^2}-\nu_i)\]
 Considering $\mu_i$ is sufficiently small $\forall i, 1 \le i \le s$,
 \[\ge 2h+3s+k-\sum_{r_i \in S_1}10\eta_i\]
 \[= 2h+3s+k-\sum_{r_i \in S_1}10^{-\lambda}r_i\]
 \[\ge 2h+3s+k-10^{-\lambda}R+10^{-\lambda}\]
 \[>7(3s+(s-1)\frac{25}{4})+\frac{294}{221}h\]
 \[=7d_2(t_1t'_1)\]
 If all basic edges are present in a gadget, any more edge can not be added further except two edges in each gadget.
 Let us consider the situation when the $LGG$ has edges other than the basic edges and two vertical/slanted edges for each gadget.
 Recall that unless all the basic edges in a frame are taken, the spanning ratio will not be bounded. It can be observed that if all the basic edges are selected in a frame,
 it cannot provide a path between any two vertices of a gadget with bounded dilation. Also recall that in a gadget all basic edges are necessary to keep dilation bounded and the edge $(x_i,y_i)$ is forbidden.
 Any other edge in a gadget (except two vertical/slanted edges) would conflict with some basic edge. Therefore, no other edge can be selected in a gadget.
 If there is an edge like $(x_i,y_{i1})$ or $(x_i,z_i)$ then
 it will not be possible to have the basic edge $(x_{i1},z_i)$ and for this pair of points the spanning ratio will not be bounded by 7. The frame above a
 gadget (an explicit cap above the first gadget) will forbid direct edge $(x_i,y_i)$ and though it will give an alternative path, it would have spanning ratio above 7.
 It can also be observed that the edges across any two distinct gadgets, two distinct frames, or any other
 edge from a gadget vertex to a frame vertex cannot exist due to conflict with a basic edge or a vertical/slanted edge.
 Though in a frame, near a junction of vertex chains there can be an edge across two different chains,
 that does not cause any problem since it does not provide a shorter path than the existing shortest path across gadgets ($(x_{i+1},x'_i,x_{i_2})$ and $(y_{i+1},y'_i,y_{i_2})$). In a path
 $(x_{i_1},x_{i_2},x'_{i-1})$ it is possible to take a shorter route missing $x_{i_2}$ and taking an edge between immediate neighbors of $x_{i_2}$ along the lines $\overline{x_{i_2}x'_i}$ and $\overline{x_{i_2}x_{i_1}}$.
 Recall that auxiliary vertices are placed very close to $x_{i_2}$ on these lines ensuring that total saving
 can be only $O(\epsilon_1)$ which is not sufficient to improve the dilation (refer calculations above). Similar arguments can be applied on the symmetric other side of the gadget when the path $(y_{i_1},y_{i_2},y'_{i-1})$ can miss $y_{i_2}$.
 There is a total four long arms. Vertex closing will ensure that no edge will be possible
 across them or to any other vertex. Thus, all the $LGGs$ on $V$ would have dilation $> 7$ leading to a contradiction.
\end{proof}

\begin{theorem}
 Given a point set $P$, it is NP-hard to find whether there exists an $LGG$ with dilation less than or equal to a given value $k$.
\end{theorem}
\begin{proof}
The proof can be inferred by Lemma~\ref{one} and Lemma~\ref{two}.
\end{proof}

\subsection{Combinatorial bounds on the dilation} \label{conn}
In this section we study combinatorial bounds for dilation on $LGGs$. A pointset giving lower bound of $\Omega(\sqrt{n})$ for dilation on the Gabriel Graphs was proposed in~\cite{beta}. The same pointset (shown in Figure~\ref{dfig1})
can also be used to show that any $LGG$ on this pointset has dilation of $\Omega(\sqrt{n})$. This bound matches with the known upper bound of $O(\sqrt{n})$
for the dilation of the Gabriel graphs (a Gabriel graph is also an $LGG$) on any pointset~\cite{beta}. We also show that for some pointsets where the Gabriel graph has dilation of $\Omega(\sqrt{n})$, there exists an
$LGG$ with dilation of $O(1)$.

 \begin{lemma} \label{dl}
  There exists a point set over which any LGG has dilation $\Omega(\sqrt{n})$.
 \end{lemma} 
\begin{proof}
In the lower bound construction there is a stack of horizontal lines with two points on each line. Going upwards, distance between these points decreases monotonically as
 shown in Figure~\ref{dfig1}.
 Let $l_1, l_2, \ldots, l_\frac{n}{2}$ be the horizontal lines. The abscissa of points on $l_i$ are $\frac{i-1}{n}$ and $r-\frac{i-1}{n}$ respectively where $r \ge \frac{1}{2n} + \frac{3}{2}$.
 Two horizontal lines are separated vertically by a distance of $\frac{1}{\sqrt{n}}$. Let $P = \{p_1,p_2,\ldots, p_n\}$ be the set of points and the points
 on line $l_i$ are numbered as $p_{2i-1}$ and $p_{2i}$ respectively. In this structure points can also be views as placed on two slanted lines. On any slanted line two adjacent points have distance $\theta(\frac{1}{\sqrt{n}})$. 
 If an $LGG$ does not have an edge connecting a pair of adjacent points on a slanted line the spanning ratio for this pair will be $\Omega(\sqrt{n})$.
 Now, let us consider the edges with end points on different slanted lines. Apart from the edge $(p_{\frac{n}{2}}, p_{{\frac{n}{2}}+1})$, all such edges
 would necessarily conflict with at least one of the edge between two adjacent points on a slanted line. Thus, such edges cannot be taken. 
 Now, consider the two vertices $p_1, p_2$ on $l_1$. The only path between $p_1$ and $p_2$ passes through all the points and has path length 
 $\Omega(\sqrt{n})$. Thus, the spanning ratio of any $LGG$ on this point set is $\Omega(\sqrt{n})$.
\end{proof}
 
 \begin{figure}[ht]
 \begin{minipage}[b]{0.4\linewidth}
 \centering
 \input{dilll.pstex_t}
 \caption{Point set where $LGG$ has dilation $\Omega(\sqrt{n})$}
 \label{dfig1}
 \end{minipage}
 \hspace{0.5cm}
 \begin{minipage}[b]{0.4\linewidth}
 \centering
 \input{dill0.pstex_t}
 \caption{Point set where $LGG$ has better dilation than Gabriel graph}
 \label{dfig2}
 \end{minipage}
 \end{figure}
 
 \begin{lemma} \label{adilation}
  There exists a point set $P$  such that the Gabriel Graph on $P$ has dilation $\Omega({\sqrt{n}})$ whereas there exists an $LGG$ on $P$ with dilation $O(1)$.
 \end{lemma}
 \begin{proof}
 There are points on two slanted lines as described in Figure~\ref{dfig1} (used for construction in Lemma~\ref{dl}) and additionally there is also a point between each successive pair of points 
 on a slanted line placed at the exterior side of it as shown in Figure~\ref{dfig2}. 
 While the Gabriel graph cannot have edges across the slanted lines (except the edge across the highest points on both the lines), these additional points can provide an alternate path
 to an $LGG$ instead of a path along the slanted line. Hence an $LGG$ can have edges across the two slanted lines. Thus, while the Gabriel Graph has only dark edges as shown in Figure~\ref{dfig2}, an $LGG$ can have 
 the dotted edges along with dark edges. This $LGG$ has $O(1)$ dilation, while the dilation of the Gabriel Graph is $\Omega({\sqrt{n}})$.
 \end{proof}
\section{Verification Algorithm for $LGG$}
Given a geometric graph $G = (V,E)$, let us consider the problem of deciding whether $G$ is a valid $LGG$.
It has to be verified that no two edges conflict with each other.\\
\begin{figure}[!h] 
\centering
\input{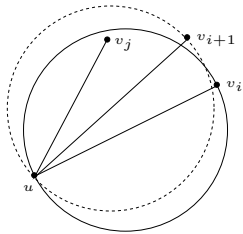}
\caption{Checking for conflicts in an $LGG$}
\label{figver}
\end{figure}
For any $u \in V$, let  $\mathcal{L}_u$ be a circular list storing all neighbors of vertex $u$ in counterclockwise order.
$G$ is a valid $LGG$ if edges from a vertex $u$ to any two consecutive members in $\mathcal{L}_u$ do not conflict with each other $\forall u \in V$. This claim follows directly from the Lemma stated below. 

\begin{lemma}\label{alglem}
Let $u$ be any vertex in $G$ and $\mathcal{L}_u = \{v_1,v_2,\ldots,v_l\}$. If edges $(u,v_i)$ and $(u,v_j)$ conflict with each other such that $i \le j - 2$, then there exist a '$k$' such that $i \le k \le j-1$
and the edge $(u,v_k)$ conflicts with the edge  $(u,v_{k+1})$.
\end{lemma}
\begin{proof}
We give a proof by contradiction.
Assume that the edges $(u,v_i)$ and $(u,v_j)$ conflict with each other and $(u,v_k)$ does not conflict with $(u,v_{k+1})$ for any $k$, s.t $i \le k < j$. Let us assume w.l.o.g. that
$(u,v_i)$ and $(u,v_j)$ are the closest pair of conflicting and non-successive edges s.t.  $i \le j - 2$, i.e. if two edges $(u,v_i')$ and $(u,v_j')$ conflict with each other and $i \le i' < j' \le j$ then $j' = i'+1$.
Since $(u,v_i)$ and $(u,v_j)$ conflict with each other, let us assume w.l.o.g. that $v_j$ lies within the circle with diameter $\overline{uv_i}$ as shown in Figure~\ref{figver}.
By assumption $(u,v_i)$ and $(u,v_{i+1})$ do not conflict, so $v_{i+1}$ must lie outside this circle and similarly $v_i$ will lie outside the circle with diameter $\overline{uv_{i+1}}$.
Recall that two circles can intersect only at two points. Now it can be trivially observed that the circle with diameter $\overline{uv_{i+1}}$ will contain $v_j$. Thus, $(u,v_{i+1})$ and
$(u,v_j)$ conflict with each other. This implies that either $(u,v_i)$ and $(u,v_j)$ are not the closest pair of conflicting edges or $(u,v_{i+1})$ and $(u,v_j)$ are successive edges
and they do conflict with each other. In either case we have a contradiction of the original assumption.
\end{proof}

The argument above directly implies a verification algorithm for $LGG$. It involves computing $\mathcal{L}_u, \forall u \in V$ that can be done by angular sorting of the neighbors of each vertex. It can be implemented in $O(|E|\log |V|)$ time.
Scanning each vertex $u$ and verifying that edges to two consecutive members in $\mathcal{L}_u$ do not conflict takes $O(|V|+|E|)$ time. Therefore, this algorithm has time complexity of $O(|E| \log |V| + |V|)$.

\section{Concludeing Remarks}
In this paper, we have introduced \emph{Generalized Locally Gabriel Graphs} in the context when certain pair of vertices may not induce edges irrespective of geometric proximity and initiated the study of computing an
edge maximum $GLGG$. We showed the problem to be NP-hard and also APX-hard. An interesting problem is to design a polynomial time approximation algorithm for the edge maximum $GLGG$ problem.
An interesting problem is to determine whether  an edge maximum $LGG$ can be computed in polynomial time for a given point set. We have shown that computing an $LGG$ minimizing dilation in NP-hard for a given point set.
A natural question is to study this problem for approximability. Another interesting set of questions is to compute $LGGs$ that optimize other parameters like maximum independent set, chromatic number etc.


\end{document}